\documentclass[runningheads]{llncs}

\usepackage[utf8]{inputenc}
\usepackage[english]{babel}
\usepackage{amsmath}
\usepackage{subfig}
\usepackage{soul}
\usepackage{amsfonts}
\usepackage{amssymb}
\usepackage{paralist}
\usepackage{xspace}
\usepackage{graphicx}

\let\doendproof\endproof
\renewcommand\endproof{~\hfill$\qed$\doendproof}

\begin{document}

\title{The QuaSEFE Problem\thanks{Work started at  Dagstuhl Seminar 19092, ``Beyond-Planar Graphs: Combinatorics, Models and Algorithms.''
Research supported by MIUR Project ``MODE'' under PRIN 20157EFM5C, by MIUR Project ``AHeAD'' under PRIN 20174LF3T8, by Roma Tre University Azione 4 Project ``GeoView'', by DFG grant Ka812/17-1, by NSF under grants CCF-1740858 and
CCF-1712119, and by SNSF Project
  200021E-171681.}}

\author{Patrizio Angelini\inst{1}\orcidID{0000-0002-7602-1524} \and Henry F\"orster\inst{1}\orcidID{0000-0002-1441-4189} \and Michael Hoffmann\inst{2}\orcidID{0000-0001-5307-7106} \and Michael Kaufmann\inst{1}\orcidID{0000-0001-9186-3538} \and Stephen Kobourov\inst{3}\orcidID{0000-0002-0477-2724} \and Giuseppe Liotta\inst{4} \and Maurizio Patrignani\inst{5}}

\institute{University of T\"ubingen, T\"ubingen \and 
ETH Z\"urich, Z\"urich \and
University of Arizona, Tucson \and
University of Perugia, Perugia \and
University Roma Tre, Rome}

\authorrunning{P. Angelini et al.}

\titlerunning{The QuaSEFE Problem}

\newcommand{\MK}{\texttt{QuaSEFE}\xspace}
\newcommand{\QSEFE}{\texttt{QuaSEFE}\xspace}
\newcommand{\BL}{\texttt{Beyond-Union}\xspace}
\newcommand{\GSE}{\texttt{GSE}\xspace}
\newcommand{\SEFE}{\texttt{SEFE}\xspace}
\newcommand{\PEP}{\texttt{PEP}\xspace}
\newcommand{\sunflowerSetting}{sunflower setting\xspace}

\maketitle
\begin{abstract}
We initiate the study of Simultaneous Graph Embedding with Fixed Edges in the beyond planarity framework.  
In the \QSEFE problem, we allow edge crossings, as long as each graph individually is drawn quasiplanar, that is, no three edges pairwise cross. 
We show that a triple consisting of two planar graphs and a tree admit a \QSEFE. This result also implies that a  pair consisting of a 1-planar graph and a planar graph admits a \QSEFE. We show several other positive results for triples of planar graphs, in which certain structural properties for their common subgraphs are fulfilled. For the case in which simplicity is also required, we give a triple consisting of two quasiplanar graphs and a star that does not admit a \QSEFE. Moreover, in contrast to the planar \SEFE problem, we show that it is not always possible to obtain a \QSEFE for two matchings if the quasiplanar drawing of one matching is fixed.

\end{abstract}

\section{Introduction}

Simultaneous Graph Embedding is a family of problems where one is given a set of graphs $G_1, \dots, G_k$ with shared vertex set $V$ and is required to produce drawings $\Gamma_1, \dots, \Gamma_k$ of them, each satisfying certain readability properties, so that each vertex has the same position in every $\Gamma_i$. The readability property that is usually pursued is the planarity of the drawing, and a large body of research has been devoted to establish the complexity of the corresponding decision problem, or to determine whether such embeddings always exist, given the number and the types of the graphs; for a survey refer to~\cite{BKR13}.

These problems have been studied both from a geometric (\emph{Geometric Simultaneous Embedding} - \GSE)~\cite{AGKN12,EGJPSS07} and from a topological point of view (\emph{Simultaneous Embedding with Fixed Edges} - \SEFE)~\cite{BR16,BCDEEIKLM07,DBLP:conf/gd/Frati06}. In particular, in \GSE the edges are straight-line segments, while in \SEFE they are topological curves, but the edges shared between two graphs $G_i$ and $G_j$ have to be drawn in the same way in $\Gamma_i$ and $\Gamma_j$. Unless otherwise specified, we focus on the topological setting.

We study a relaxation of the \SEFE problem, where the 
graphs can be drawn 
with edge crossings. However, we prohibit certain crossing configurations in the drawings $\Gamma_1, \dots, \Gamma_k$, to guarantee their readability, i.e., we require that they satisfy the conditions of a graph class in the area of \emph{beyond-planarity}; see~\cite{DBLP:journals/csur/DidimoLM19} for a survey on this topic. 
We initiate this study with the class of \emph{quasiplanar} graphs~\cite{AT07,AgarwalAPPS97,FoxPS13}, by requiring that no $\Gamma_i$ contains three mutually crossing edges.

\begin{definition}[\QSEFE]
Given a set of graphs $G_1=(V,E_1), \ldots, G_k=(V,E_k)$ with shared vertex set $V$, we say that $\langle G_1, \ldots, G_k\rangle$ admits a \QSEFE if there exist quasiplanar drawings $\Gamma_1, \dots, \Gamma_k$ of $G_1, \ldots, G_k$, respectively, so that each vertex of $V$ has the same position in every $\Gamma_i$ and each edge shared between two graphs $G_i$ and $G_j$ is drawn in the same way in $\Gamma_i$ and $\Gamma_j$. 
Further, the \QSEFE problem asks whether an instance $\langle G_1, \ldots, G_k \rangle$ admits a \QSEFE.
\end{definition}



It may be worth mentioning that the problem of computing quasiplanar simultaneous embeddings of graph pairs has been studied in the geometric setting~\cite{DBLP:journals/cj/GiacomoDLMW15,DBLP:journals/ipl/DidimoKLOS12}. Also, simultaneous embeddings have been considered in relation to another beyond-planarity geometric graph class, namely \emph{RAC graphs}~\cite{ArgyriouBKS13,DBLP:journals/jgaa/BekosDKW16,DBLP:journals/tcs/EvansLM16,DBLP:journals/jgaa/Grilli18}.

We prove in Section~\ref{sec:algorithms} that any triple of two planar graphs and a tree admits a \QSEFE, which also implies that any pair consisting of a $1$-planar graph\footnote[1]{A graph is $k$-planar if it admits a drawing where each edge has at most $k$ crossings.} and a planar graph admits a \QSEFE. Recall that, for the original \SEFE problem, there exist even negative instances composed of two outerplanar graphs~\cite{DBLP:conf/gd/Frati06}. Further, we investigate triples of planar graphs in which the common subgraphs have specific structural properties. Finally, we show negative results in more specialized settings in Section~\ref{sec:negative} and conclude
 with open problems in Section~\ref{sec:openProblems}.






\section{Sufficient Conditions for \QSEFE{s}}\label{sec:algorithms}

In this section, we provide several sufficient conditions for the existence of a \QSEFE, mainly focusing on instances composed of three planar graphs $G_1$, $G_2$, and $G_3$. 
We start with a theorem relating the existence of a \SEFE of two of the input graphs to the existence of a \QSEFE of the three input graphs.

\begin{theorem}
Let $G_1=(V,E_1)$, $G_2=(V,E_2)$, and $G_3=(V,E_3)$ be planar graphs with shared vertex set $V$. If $\langle G_1\setminus G_3,G_2\setminus G_3 \rangle$ admits a \SEFE, then $\langle G_1, G_2, G_3 \rangle$ admits a \QSEFE, in which the drawing of $G_3$ is planar.
\label{thm:quaSEFEFromSEFE}
\end{theorem}

\begin{proof}
First construct a \SEFE of $\langle G_1\setminus G_3,G_2\setminus G_3 \rangle$, and then construct a planar drawing of $G_3$, whose vertices have already been placed, but whose edges have not been drawn yet, using the algorithm by Pach and Wenger~\cite{PW01}. 

The drawing of $G_3$ is planar, by construction. The drawing of $G_1$ is quasiplanar, as it is partitioned into two subgraphs, $G_1\setminus G_3$ and $G_1\cap G_3$, each of which is drawn planar. Analogously, the drawing of $G_2$ is quasiplanar.
\end{proof}

Since every pair composed of a planar graph and a tree admits a \SEFE~\cite{DBLP:conf/gd/Frati06}, we derive from Theorem~\ref{thm:quaSEFEFromSEFE} the following positive result for the \QSEFE problem.

\begin{corollary}
Let $G_1=(V,E_1)$ and $G_3=(V,E_3)$ be planar graphs and $T_2=(V,E_2)$ be a tree with shared vertex set $V$. Then $\langle G_1,T_2,G_3 \rangle$ admits a \QSEFE, in which the drawing of $G_3$ is planar.
\label{cor:quaSEFEPPT}
\end{corollary}

Corollary~\ref{cor:quaSEFEPPT} already shows that allowing quasiplanarity significantly enlarges the set of positive instances. 
We further strengthen this result, by additionally guaranteeing that even the tree is drawn planar. For this, we use a result on the \emph{partially embedded planarity}~\cite{ABFJKPR15} problem (\PEP): 
Given a planar graph $G$, a subgraph $H$ of $G$, and a planar embedding $\mathcal{H}$ of $H$, is there a planar embedding of $G$ whose restriction to $H$ coincides with $\mathcal{H}$?
In particular, we will exploit the following characterization, which is the core of a linear-time algorithm for \PEP.

\begin{lemma}[\cite{ABFJKPR15}] Let $(G,H,\mathcal{H})$ be an instance of \PEP. A planar embedding $\mathcal{G}$ of $G$ is a solution for $(G,H,\mathcal{H})$ if and only if the following conditions hold:
\begin{inparaenum}[\bf ({C.}1)]
    \item \label{condition:1PEP} for every vertex $v \in V$, the edges incident to $v$ in $H$ appear in the same cyclic order in the rotation schemes of $v$ in $\mathcal{H}$ and in $\mathcal{G}$; and
    \item \label{condition:2PEP} for every cycle $C$ of $H$, and for every vertex $v$ of $H\setminus C$, we have that $v$ lies in the interior of $C$ in $\mathcal{G}$ if and only if it lies in the interior of $C$ in $\mathcal{H}$.
\end{inparaenum}
\label{lem:partialEmbedding}
\end{lemma}


\begin{theorem}
Let $G_1=(V,E_1)$ and $G_3=(V,E_3)$ be planar graphs and $T_2=(V,E_2)$ be a tree with shared vertex set $V$. Then $\langle G_1,T_2,G_3 \rangle$ admits a \QSEFE, in which the drawings of $G_1$ and $T_2$ are planar.
\label{thm:mkTwoPlanarOneTree}
\end{theorem}

\begin{proof}


Consider planar embeddings $\mathcal{G}_1$ and $\mathcal{G}_3^*$ of $G_1$ and $G_3 \setminus G_1$, respectively.
We draw $G_1$ according to $\mathcal{G}_1$. This fixes the embedding of the subgraph $T_2 \cap G_1$ of $T_2$, thus resulting in an instance of the \PEP problem. 
Since $T_2$ is acyclic, Condition~C.\ref{condition:2PEP} of Lemma~\ref{lem:partialEmbedding} is trivially fulfilled. Also, since every rotation scheme of $T_2$ is planar, we choose for the edges of $(T_2 \cap G_3) \setminus G_1$ an order compatible with $\mathcal{G}_3^*$, still satisfying Condition~C.\ref{condition:1PEP}.
Finally, we draw the remaining edges of~$G_3$ by considering the instance of~\PEP defined by its embedded subgraph $(T_2 \cap G_3) \setminus G_1$. Condition~C.\ref{condition:2PEP} is trivially satisfied, and Condition~C.\ref{condition:1PEP} is satisfied by construction, if we add the edges of~$G_3$ according to $\mathcal{G}_3^*$. 
Since crossings edges of the same graph belong to $G_3 \setminus G_1$ and $G_3 \cap G_1$, the drawing of $G_3$ is quasiplanar.
\end{proof}

The additional property guaranteed by Theorem~\ref{thm:mkTwoPlanarOneTree} is crucial to infer the first result in the simultaneous embedding setting for a class of beyond-planar graphs.

\begin{theorem}
Let $G_1=(V,E_1)$ be a $1$-planar graph and $G_2=(V,E_2)$ be a planar graph. Then $\langle G_1,G_2 \rangle$ admits a \QSEFE. 
\label{th:1planarPlanar}
\end{theorem}

\begin{proof}
As $G_1$ is $1$-planar, it is the union of a planar graph $G_1'$ and a forest $F_1$~\cite{A14}. We augment $F_1$ to a tree $T_1$. 
By Theorem~\ref{thm:mkTwoPlanarOneTree}, there is a \QSEFE of $\langle G_1', T_1, G_2 \rangle$ where $G_1'$ and $T_1$ are drawn planar. Thus, $G_1$ is drawn quasiplanar.
\end{proof}

We now study properties of the subgraphs induced by the edges that belong to one, to two, or to all the input graphs. We denote by $H_i$ the subgraph induced by the edges only in $G_i$; by $H_{i,j}$ the subgraph induced by the edges only in $G_i$ and $G_j$; and by $H$ the subgraph induced by the edges in all graphs; see Fig.~\ref{fig:intersections}.


The following two corollaries of Theorem~\ref{thm:quaSEFEFromSEFE} list sufficient conditions for $G_1 \setminus G_3$ and $G_2\setminus G_3$ to have a \SEFE. In the first case, $H_{1,2}$ has a unique embedding, which fulfills the conditions of Lemma~\ref{lem:partialEmbedding} with respect to any planar 
embedding of $G_1$ and of $G_2$. In the second case, this is because $G_1\setminus G_3$ is a subgraph of $G_2 \setminus G_3$.

\begin{corollary}
Let $G_1=(V,E_1)$, $G_2=(V,E_2)$, $G_3=(V,E_3)$ be planar graphs with shared vertex set $V$. If $H_{1,2}$ is acyclic and has maximum degree $2$, then $\langle G_1, G_2, G_3 \rangle$ admits a \QSEFE.
\label{cor:G1G2Small}
\end{corollary}


\begin{corollary}
Let $G_1=(V,E_1)$, $G_2=(V,E_2)$, $G_3=(V,E_3)$ be planar graphs with shared vertex set $V$. If $H_1 = \emptyset$, then $\langle G_1, G_2, G_3 \rangle$ admits a \QSEFE.
\label{cor:G1Empty}
\end{corollary}


Contrary to the previous corollaries, Theorem~\ref{thm:quaSEFEFromSEFE} has no implication for the graph $H$, as there are instances with $H = \emptyset$ where no pair of graphs has a \SEFE. However, we show that a simple structure of $H$ is still sufficient for a \QSEFE.

\begin{theorem}
Let $G_1=(V,E_1)$, $G_2=(V,E_2)$, $G_3=(V,E_3)$ be planar graphs with shared vertex set $V$. If $H$ has a planar embedding that can be extended to a planar embedding $\mathcal{G}_i$ of each graph $G_i$, then $\langle G_1,G_2, G_3 \rangle$ admits a \QSEFE.
\label{thm:intersectionExtendable}
\end{theorem}

\begin{proof} 
We draw the graph $G_1\setminus H_{1,3} = H_1 \cup H_{1,2} \cup H$ with embedding $\mathcal{G}_1$, the graph $G_2 \setminus H_{1,2} = H_2 \cup H_{2,3} \cup H$ with embedding $\mathcal{G}_2$, and the graph $G_3\setminus H_{2,3} = H_3 \cup H_{1,3} \cup H$ with embedding $\mathcal{G}_3$. Then, the edges of $G_1$ are partitioned into two sets, one belonging to $G_1\setminus H_{1,3}$ and one to $G_3\setminus H_{2,3}$, each of which is drawn planar. As the same holds for the edges of $G_2$ and $G_3$, the statement follows.
\end{proof}

\begin{corollary}
Let $G_1=(V,E_1)$, $G_2=(V,E_2)$, $G_3=(V,E_3)$ be planar graphs with shared vertex set $V$. If $H$ is acyclic and has maximum degree $2$, then $\langle G_1,G_2, G_3 \rangle$ admits a \QSEFE.
\label{cor:intersectionSmall}
\end{corollary}


The above discussion shows that, if one of the seven subgraphs in Fig.~\ref{fig:intersections} is empty, or has a sufficiently simple structure, $\langle G_1, G_2, G_3\rangle$ admits a \QSEFE. Most notably, this is always the case in the \emph{sunflower} setting~\cite{DBLP:journals/tcs/AngeliniLN15,DBLP:journals/jgaa/HaeuplerJL13,DBLP:journals/jgaa/Schaefer13}, in which every edge belongs either to a single graph or to all graphs, i.e., $H_{1,2}=H_{1,3}=H_{2,3}=\emptyset$. We extend this result to any set of planar graphs. We remark that \SEFE is NP-complete in the sunflower setting for three planar graphs~\cite{DBLP:journals/tcs/AngeliniLN15,DBLP:journals/jgaa/Schaefer13}.

\begin{theorem}
Let $G_1=(V,E_1), \ldots, G_k=(V,E_k)$ be planar graphs with shared vertex set $V$ in the \sunflowerSetting. Then $\langle G_1, \ldots,G_k \rangle$ admits a \QSEFE.
\label{thm:sunflower}
\end{theorem}

\begin{proof}
Let $H$ be the graph induced by the edges belonging to all graphs. We independently draw planar the graph $H$ and every subgraph $G_i \setminus H$, for $i=1,\ldots,k$. This guarantees that each $G_i$ is drawn quasiplanar.
\end{proof}

\begin{figure}[tb]
    \centering
    \subfloat[\label{fig:intersections}]
	{\includegraphics[page=1,width=0.23\textwidth]{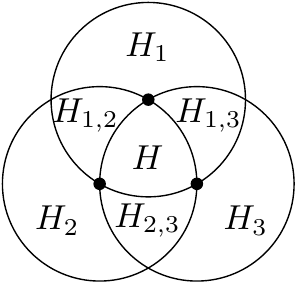}}
	\hfil
	\subfloat[\label{fig:k11-5edges}]
	{\includegraphics[width=0.27\textwidth]{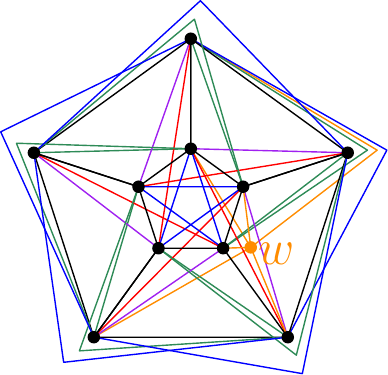}}
	\hfil
	\subfloat[\label{fig:escape-game}]
	{\includegraphics[page=4,width=0.45\textwidth]{figures/escape-game}}
    \caption{
    (a) Subgraphs induced by the edges in one, two, or three graphs.
    (b) A simple quasiplanar drawing of  $Q_1$ in Theorem~\ref{thm:simpleNegativeInstance}, obtained by adding $w$ to the drawing of $K_{10}$ by Brandenburg~\cite{brandenburg16}. (c)~Theorem~\ref{lem:escape-game}: Edge $(v_{18},v_{20})$ crosses either all dotted blue or all dashed red edges, making $(v_5,v_6)$ and $(v_7,v_8)$ uncrossable.}
\end{figure}

We remark that all our proofs are constructive. Moreover, the corresponding algorithms run in linear time, as they exploit linear-time algorithms for constructing planar embeddings of graphs~\cite{DBLP:journals/jacm/HopcroftT74}, for extending their partial embeddings~\cite{ABFJKPR15}, and for partitioning $1$-planar graphs into planar graphs and forests~\cite{A14}.

\section{Counterexamples for \QSEFE}
\label{sec:negative}

In this section we complement our positive results, by providing negative instances of the \QSEFE problem in two specific settings.
We start with a negative result about the existence of a \emph{simple} \QSEFE for two quasiplanar graphs and one star. Here \emph{simple} means that a pair of independent edges in the same graph is allowed to cross at most once and a pair of adjacent edges in the same graph is not allowed to cross. Note that our algorithms in Section~\ref{sec:algorithms} may produce non-simple drawings. Also, the maximum number of edges in a quasiplanar graph with $n$ vertices depends on whether simplicity is required or not~\cite{AT07}.

\begin{theorem}
\label{thm:simpleNegativeInstance}
There exist two quasiplanar graphs $Q_1=(V,E_1)$, $Q_2=(V,E_2)$ and a star $S_3=(V,E_3)$ with shared vertex set $V$ such that $\langle Q_1, Q_2, S_3 \rangle$ does not admit a simple \QSEFE.
\end{theorem}

\begin{proof}
Let $V = \{v_1,\ldots,v_{10},w\}$ and let $E_{10}$ be the edges of the complete graph on $V\setminus \{w\}$. Further, let $E_1 = E_{10} \cup \{(w,v_1),\ldots,(w,v_6)\}$, let $E_2 = E_{10} \cup \{(w,v_7)\}$, and let $E_3=\{(w,v_1),\ldots,(w,v_{10}\}$. By construction, $S_3$ is the star on all eleven vertices with center $w$, while Fig.~\ref{fig:k11-5edges} shows that there is a simple quasiplanar drawing of $Q_1$ (and of $Q_2$, which is a subgraph of $Q_1$, up to vertex relabeling). 

Suppose that $\langle Q_1, Q_2, S_3 \rangle$ has a simple \QSEFE, and let $\Gamma_{1,2}$ be the drawing of the union of $Q_1$ and $Q_2$ that is part of it. Since the union of $Q_1$ and $Q_2$ has $52$ edges, which exceeds the upper bound of $6.5n-20$ edges in a simple quasiplanar graph~\cite{AT07}, $\Gamma_{1,2}$ is not simple or not quasiplanar. 
Since $(w,v_7)$ is the only edge in $\Gamma_{1,2}$ that is not in $Q_1$,
edge $(w,v_7)$ is involved in every crossing violating simplicity or quasiplanarity. Analogously, one of $(w,v_1), \dots, (w,v_6)$, say $(w,v_1)$, is involved in every crossing violating simplicity or quasiplanarity; in particular, $(w,v_1)$ crosses $(w,v_7)$. Since both $(w,v_1)$ and $(w,v_7)$ belong to $S_3$, the drawing of $S_3$ that is part of the simple \QSEFE is not simple, a contradiction.
\end{proof}

The second special setting is the one in which one of the input graphs is already drawn in a quasiplanar way, and the goal is to draw the other input graphs so that the resulting simultaneous drawing is a \QSEFE. This setting is motivated by the natural approach, for an instance $\langle G_1, \ldots, G_k \rangle$, of first constructing a solution for $\langle G_1, \ldots, G_{k-1} \rangle$ and then adding the remaining edges of~$G_k$. Note that, since the drawing of the first graph partially fixes a drawing of the second graph, this can be seen as a version of the \PEP problem for quasiplanarity. 

For the original \SEFE problem, this setting always has a solution when the graph that is already drawn (in a planar way) is a general planar graph, and the other graph is a tree~\cite{DBLP:conf/gd/Frati06}.
In a surprising contrast, 
we construct negative instances for the \QSEFE problem that are composed of two matchings only.


\begin{theorem}\label{lem:escape-game}
Let $M_1=(V,E_1)$ and $M_2=(V,E_2)$ be two matchings on the same vertex set $V$ and let $\Gamma_1$ be a quasiplanar drawing of $M_1$. Instance $\langle M_1, M_2\rangle$ does not always admit a \QSEFE in which the drawing of $M_1$ is $\Gamma_1$.
\end{theorem}
\begin{proof}
First recall that the edges in $E_1 \cap E_2$ have to be drawn in the quasiplanar drawing~$\Gamma_2$ of $G_2$ as they are in~$\Gamma_1$. 
Consider the quasiplanar drawing $\Gamma_1$ of the matching $(v_{2i-1},v_{2i})$, with $i=1, \dots, 10$, in Fig.~\ref{fig:escape-game}, 
and let $E_2$ contain the edges $(v_{17},v_{19})$ and $(v_{18},v_{20})$.
Since $v_{17}$ is enclosed in a region bounded by the crossing edges $(v_1,v_2)$ and $(v_3,v_4)$, in any quasiplanar drawing of $M_2$ edge $(v_{17},v_{19})$ crosses exactly one of $(v_1,v_2)$ and $(v_3,v_4)$. In the first case, 
$(v_{17},v_{19})$ crosses also $(v_{13},v_{14})$ and $(v_{15},v_{16})$ (dotted blue). In the second case, $(v_{17},v_{19})$ crosses also $(v_{9},v_{10})$ and $(v_{11},v_{12})$
(dashed red). In both cases, $(v_{5},v_{6})$ and $(v_{7},v_{8})$ cannot be crossed, and thus $(v_{17},v_{19})$ cannot be drawn so that $\Gamma_2$ is quasiplanar.
\end{proof}


\section{Conclusions and Open Problems}
\label{sec:openProblems}

We initiated the study of simultaneous embeddability in the beyond planar setting, which is a fertile and almost unexplored research direction that promises to significantly enlarge the families of representable graphs when compared with the planar setting. We conclude the paper by listing a few open problems.

\begin{itemize}
\item A natural question is whether two $1$-planar graphs, a quasiplanar graph and a matching, three outerplanar graphs, or four paths admit a \QSEFE.
All our algorithms construct drawings with a stronger property than quasiplanarity, namely that they are composed of two sets of planar edges. Exploiting quasiplanarity in full generality may lead to further positive results.
\item Motivated by Theorem~\ref{thm:simpleNegativeInstance}, we ask whether some of the constructions presented in Section~\ref{sec:algorithms} can be modified to guarantee the simplicity of the drawings.
\item Another intriguing direction is to determine the computational complexity of the \QSEFE problem, both in its general version and in the two restrictions studied in Section~\ref{sec:negative}. In particular, the setting in which one of the graphs is already drawn can be considered as a quasiplanar version of the \PEP problem, which is known to be linear-time solvable in the planar case~\cite{ABFJKPR15}. 
\item Extend the study to other beyond-planarity classes. For example, do any two planar graphs admit a $k$-planar \SEFE for some constant $k$?
\end{itemize}

\clearpage

\bibliographystyle{splncs04}
\bibliography{references}

\begin{thebibliography}{10}
\providecommand{\url}[1]{\texttt{#1}}
\providecommand{\urlprefix}{URL }
\providecommand{\doi}[1]{https://doi.org/#1}

\bibitem{A14}
Ackerman, E.: A note on 1-planar graphs. Discrete Applied Mathematics
  \textbf{175},  104--108 (2014). \doi{10.1016/j.dam.2014.05.025}

\bibitem{AT07}
Ackerman, E., Tardos, G.: On the maximum number of edges in quasi-planar
  graphs. Journal of Combinatorial Theory, Series A  \textbf{114}(3),  563 --
  571 (2007). \doi{10.1016/j.jcta.2006.08.002}

\bibitem{AgarwalAPPS97}
Agarwal, P.K., Aronov, B., Pach, J., Pollack, R., Sharir, M.: Quasi-planar
  graphs have a linear number of edges. Combinatorica  \textbf{17}(1), ~1--9
  (1997). \doi{10.1007/BF01196127}

\bibitem{DBLP:journals/tcs/AngeliniLN15}
Angelini, P., {Da Lozzo}, G., Neuwirth, D.: Advancements on {SEFE} and
  partitioned book embedding problems. Theor. Comput. Sci.  \textbf{575},
  71--89 (2015). \doi{10.1016/j.tcs.2014.11.016}

\bibitem{ABFJKPR15}
Angelini, P., Di~Battista, G., Frati, F., Jel\'{\i}nek, V., Kratochv\'{\i}l,
  J., Patrignani, M., Rutter, I.: Testing planarity of partially embedded
  graphs. ACM Trans. Algorithms  \textbf{11}(4),  32:1--32:42 (2015).
  \doi{10.1145/2629341}

\bibitem{AGKN12}
Angelini, P., Geyer, M., Kaufmann, M., Neuwirth, D.: On a tree and a path with
  no geometric simultaneous embedding. J. Graph Algorithms Appl.
  \textbf{16}(1),  37--83 (2012)

\bibitem{ArgyriouBKS13}
Argyriou, E.N., Bekos, M.A., Kaufmann, M., Symvonis, A.: Geometric {RAC}
  simultaneous drawings of graphs. J. Graph Algorithms Appl.  \textbf{17}(1),
  11--34 (2013). \doi{10.7155/jgaa.00282}

\bibitem{DBLP:journals/jgaa/BekosDKW16}
Bekos, M.A., van Dijk, T.C., Kindermann, P., Wolff, A.: Simultaneous drawing of
  planar graphs with right-angle crossings and few bends. J. Graph Algorithms
  Appl.  \textbf{20}(1),  133--158 (2016). \doi{10.7155/jgaa.00388}

\bibitem{BKR13}
Bl{\"{a}}sius, T., Kobourov, S.G., Rutter, I.: Simultaneous embedding of planar
  graphs. In: Tamassia, R. (ed.) Handbook on Graph Drawing and Visualization.,
  pp. 349--381. Chapman and Hall/CRC (2013)

\bibitem{BR16}
Bl{\"{a}}sius, T., Rutter, I.: Simultaneous {PQ}-ordering with applications to
  constrained embedding problems. {ACM} Trans. Algorithms  \textbf{12}(2),
  16:1--16:46 (2016). \doi{10.1145/2738054}

\bibitem{brandenburg16}
Brandenburg, F.J.: A simple quasi-planar drawing of {$K_{10}$}. In: Hu, Y.,
  N{\"{o}}llenburg, M. (eds.) Graph Drawing. LNCS, vol.~9801, pp. 603--604.
  Springer (2016)

\bibitem{BCDEEIKLM07}
Bra{\ss}, P., Cenek, E., Duncan, C.A., Efrat, A., Erten, C., Ismailescu, D.,
  Kobourov, S.G., Lubiw, A., Mitchell, J.S.B.: On simultaneous planar graph
  embeddings. Comput. Geom.  \textbf{36}(2),  117--130 (2007).
  \doi{10.1016/j.comgeo.2006.05.006}

\bibitem{DBLP:journals/cj/GiacomoDLMW15}
{Di Giacomo}, E., Didimo, W., Liotta, G., Meijer, H., Wismath, S.K.: Planar and
  quasi-planar simultaneous geometric embedding. Comput. J.  \textbf{58}(11),
  3126--3140 (2015). \doi{10.1093/comjnl/bxv048}

\bibitem{DBLP:journals/ipl/DidimoKLOS12}
Didimo, W., Kaufmann, M., Liotta, G., Okamoto, Y., Spillner, A.: Vertex angle
  and crossing angle resolution of leveled tree drawings. Inf. Process. Lett.
  \textbf{112}(16),  630--635 (2012). \doi{10.1016/j.ipl.2012.05.006}

\bibitem{DBLP:journals/csur/DidimoLM19}
Didimo, W., Liotta, G., Montecchiani, F.: A survey on graph drawing beyond
  planarity. {ACM} Comput. Surv.  \textbf{52}(1),  4:1--4:37 (2019).
  \doi{10.1145/3301281}

\bibitem{EGJPSS07}
Estrella{-}Balderrama, A., Gassner, E., J{\"{u}}nger, M., Percan, M., Schaefer,
  M., Schulz, M.: Simultaneous geometric graph embeddings. In: Hong, S.,
  Nishizeki, T., Quan, W. (eds.) 15th International Symposium on Graph Drawing,
  {GD} 2007. LNCS, vol.~4875, pp. 280--290. Springer (2007).
  \doi{10.1007/978-3-540-77537-9\_28}

\bibitem{DBLP:journals/tcs/EvansLM16}
Evans, W.S., Liotta, G., Montecchiani, F.: Simultaneous visibility
  representations of plane st-graphs using {L}-shapes. Theor. Comput. Sci.
  \textbf{645},  100--111 (2016). \doi{10.1016/j.tcs.2016.06.045}

\bibitem{FoxPS13}
Fox, J., Pach, J., Suk, A.: The number of edges in k-quasi-planar graphs. SIDMA
   \textbf{27}(1),  550--561 (2013). \doi{10.1137/110858586}

\bibitem{DBLP:conf/gd/Frati06}
Frati, F.: Embedding graphs simultaneously with fixed edges. In: Kaufmann, M.,
  Wagner, D. (eds.) Graph Drawing, 14th International Symposium, {GD} 2006.
  LNCS, vol.~4372, pp. 108--113. Springer (2006).
  \doi{10.1007/978-3-540-70904-6\_12}

\bibitem{DBLP:journals/jgaa/Grilli18}
Grilli, L.: On the {NP}-hardness of {GRacSim} drawing and {k-SEFE} problems. J.
  Graph Algorithms Appl.  \textbf{22}(1),  101--116 (2018).
  \doi{10.7155/jgaa.00456}

\bibitem{DBLP:journals/jgaa/HaeuplerJL13}
Haeupler, B., Jampani, K.R., Lubiw, A.: Testing simultaneous planarity when the
  common graph is 2-connected. J. Graph Algorithms Appl.  \textbf{17}(3),
  147--171 (2013). \doi{10.7155/jgaa.00289}

\bibitem{DBLP:journals/jacm/HopcroftT74}
Hopcroft, J.E., Tarjan, R.E.: Efficient planarity testing. J. {ACM}
  \textbf{21}(4),  549--568 (1974). \doi{10.1145/321850.321852}

\bibitem{PW01}
Pach, J., Wenger, R.: Embedding planar graphs at fixed vertex locations. Graphs
  and Combinatorics  \textbf{17}(4),  717--728 (2001). \doi{10.1007/PL00007258}

\bibitem{DBLP:journals/jgaa/Schaefer13}
Schaefer, M.: Toward a theory of planarity: Hanani-tutte and planarity
  variants. J. Graph Algorithms Appl.  \textbf{17}(4),  367--440 (2013).
  \doi{10.7155/jgaa.00298}

\end{thebibliography}



\end{document}